\documentclass[12pt]{article}
\usepackage{amsmath,amssymb,amsthm,mathtools}
\usepackage{amsfonts}
\usepackage[nohead]{geometry}
\usepackage[singlespacing]{setspace}
\usepackage[bottom]{footmisc}
\usepackage{indentfirst}
\usepackage{endnotes}
\usepackage{graphicx}
\usepackage{rotating}
\usepackage{enumitem}
\usepackage{microtype}
\usepackage{color} 
\newtheorem{theorem}{Theorem}

\newtheorem{claim}{Claim}

\newtheorem{definition}{Definition}

\newtheorem{lemma}{Lemma}
\newtheorem{proposition}{Proposition}
\newtheorem{remark}{Remark}
\newcommand{\argmax}{\mathop{\rm arg~max}\limits}

\newcommand{\qedo}{\hfill \rule{5pt}{10pt} \par}
\makeatletter
\def\@biblabel#1{\hspace*{-\labelsep}}
\makeatother
\usepackage{natbib}
\bibpunct[:]{(}{)}{,}{a}{}{,}

\newcommand\cites[1]{\citeauthor{#1}'s\ (\citeyear{#1})}

\allowdisplaybreaks[1]

\newcommand{\R}{\mathbb{R}}

\newcommand{\dist}{\operatorname{dist}}

\newcommand{\1}{\mathbf{1}}

\begin{document}
\title{Constrained optimal transport with applications to large matching markets\footnote{The author is grateful to Michihiro Kandori, Fuhito Kojima, Kenji Tsukada, Reo Nonaka, Yuki Tsutsui, Atsushi Iwasaki, Alex Teytelboym, and Ravi Jagadeesan, as well as to seminar participants at the Tokyo Conference on Market Design 2025 and SWET2025, for their valuable comments and suggestions. 
This work was supported by Japan Science and Technology Agency ERATO Grant Number JPMJER2301, Japan.}}
\author{Koji Yokote\footnote{College of Economics, Aoyama Gakuin University, Tokyo 150-8366, Japan. Email: koji.yokote@gmail.com}}
\maketitle

%%%%%%%Note to self
%%%1: the symbol of ``transpose'' (tenchi) is used in some places but not in other places. Let's be consistent. 
%%%2: It seems better to add some preamble before each theorem. 
%%%3 technical temrs are in bald font in some places but not in others. Let's be consistent. 

\begin{abstract}
We establish a variant of Monge--Kantorovich duality for matching problems with a continuum of agents, a finite set of alternatives, and general linear constraints. Under primal feasibility alone, both the primal and dual problems attain their optimal values, and the values coincide.
We develop two applications. First, in a large market with indivisible goods, equilibrium prices are characterized as minimizers of a finite-dimensional convex potential function, and a t{\^a}tonnement-style subgradient process converges to equilibrium prices even in the presence of complementarities. Second, in a transferable-utility matching market with couples, where the core may be empty in finite markets, dual optimizers support a nonempty core in the continuum economy. These results extend optimal-transport methods to a broader class of constrained large matching markets.
%JEL classification:  \\
%Keywords: 
\end{abstract}

\section{Introduction}

Optimal transport problems constitute a class of infinite-dimensional
linear programs. Their analysis is made tractable by the strong duality
theorem commonly referred to as Monge--Kantorovich (henceforth, MK)
duality. Optimal-transport methods have found numerous applications in
economics and have become a central tool in a substantial strand of the
empirical literature on transferable-utility matching
\citep[e.g.,][]{DupuyGalichon2014,CiscatoGalichonGousse2020,
GalichonSalanie2022Cupid,GalichonSalanie2024}; see also
\citet{Chiappori2017Matching}. In these applications, primal--dual
representations link observed matching patterns to matching surplus and
equilibrium utilities and provide methods for identification, estimation,
and computation. For broader surveys of optimal transport in economics,
see \citet{Galichon2016OT,Galichon2021Unreasonable}.

This paper establishes a variant of MK duality for environments in which
one side of the market is a continuum and the other side is finite, with
the finite side subject to general linear equality and inequality
constraints. The distribution of agent types is fixed, whereas the
aggregate assignments to the finite set of alternatives are determined
endogenously subject to these constraints. The formulation accommodates,
among other restrictions, upper and lower capacities,
minimum-utilization requirements, resource constraints linking several
alternatives, and accounting or balance conditions. Unlike a classical
transport problem with a prescribed target marginal, our model does not
require the occupancy of every finite-side alternative to be fixed in
advance. Instead, its occupancy is chosen as part of the
welfare-maximizing matching.

Our main theorem establishes strong duality and both primal and dual
attainment under mere feasibility of the primal constraints. The
finite-dimensional dual variables are the multipliers associated with the
aggregate constraints and may represent prices, taxes, subsidies, or
shadow values, depending on the application. Thus, the theorem applies to
large-market environments in which a continuum of heterogeneous agents is
assigned to finitely many goods, institutions, or bundles whose aggregate
use is institutionally constrained.

To the best of our knowledge, existing optimal-transport duality results
do not establish dual attainment for this continuum--finite model under
mere primal feasibility. The main obstacle is that the multiplier space
is not compact: strong duality alone does not exclude the possibility
that every minimizing sequence of dual variables escapes to infinity. We
derive compactness directly from the geometry of the finite-side
constraint system. Hoffman's error bound yields an exact finite penalty
for violations of the linear constraints, allowing us to restrict the
multipliers to a compact box without changing the optimal value. Sion's
minimax theorem then identifies the reduced dual problem on this compact
domain, where continuity guarantees attainment. The distinctive feature
of the argument is that compactness is obtained from the polyhedral
structure of the primal constraints, rather than from assumed
coercivity, a Slater-type interiority condition, or
application-specific bounds on prices.

This extension of MK duality may also be useful for empirical matching.
Dual potentials encode equilibrium utilities and play a central role in
the identification and computation of optimal-transport-based matching
models. Although we do not develop a new estimator in this paper, our
theorem guarantees the existence of these economically meaningful dual
objects in continuum--finite environments with general linear
constraints. It therefore provides a foundation for adapting the
existing optimal-transport toolkit to a broader class of constrained
matching markets.

We develop two applications. 
%
%First, we revisit the large market with finitely many indivisible goods studied by \citet{AzevedoWeylWhite2013}. As shown by \citet{AzevedoYokote2026}, the original proof of competitive-equilibrium existence relies on an incorrect claim that the set of allocations is compact under the $L^1$ norm. Our duality theorem provides an alternative existence proof. Relative to the repair proposed by \citet{AzevedoYokote2026}, our argument is less general in requiring the type space to be compact. In return, it characterizes equilibrium prices as minimizers of a finite-dimensional convex potential function. Supply minus aggregate demand is a subgradient of this function, so the associated subgradient algorithm has the interpretation of a Walrasian t\^atonnement process. Under standard step-size conditions, the process converges to the set of equilibrium prices without requiring substitutability and therefore permits complementarities.
%%
%%
%
First, we revisit the large market with finitely many indivisible goods studied by \citet{AzevedoWeylWhite2013}. We identify an error in their original proof of competitive-equilibrium existence: the proof relies on an incorrect claim that the set of allocations is compact under the $L^1$ norm.\footnote{Subsequent to an earlier version of
this paper, \citet{AzevedoYokote2026} developed a different repair of the existence proof
that does not rely on duality. For a comparison between our approach and that of
\citet{AzevedoYokote2026}, see Remark~\ref{remark-proof-error}.}
We recover equilibrium existence as a direct application of our duality theorem and,
in addition, show that equilibrium prices are precisely the minimizers of a
finite-dimensional convex potential function. Supply minus aggregate demand is a
subgradient of this function, so a subgradient algorithm takes the form of a Walrasian
t{\^a}tonnement process. Under standard step-size conditions, this process converges to
the set of equilibrium prices without requiring substitutability and thus accommodates
complementarities.

Second, we study a large transferable-utility matching market with
couples. Joint preferences over pairs of positions generate
complementarities, and stable matchings need not exist in finite markets
\citep{Roth1984}. Conditions under which stability is recovered in large
nontransferable-utility markets are studied by
\citet{KojimaPathakRoth2013} and
\citet{AshlagiBravermanHassidim2014}. We consider a
transferable-utility counterpart in which a couple is jointly assigned to
two hospitals and the resulting surplus may be distributed among the
couple and the hospitals. The analogous finite-market difficulty is that
the core may be empty. We represent joint hospital assignments as
finite-side alternatives and hospital capacities as linear constraints.
Our duality theorem then produces a welfare-maximizing matching together
with dual utilities and shadow values that support a core allocation.
Thus, although the core may be empty in the finite market, it is nonempty
in the continuum economy.

The joint-assignment structure in the couples application also arises in
other important allocation problems. For example, families may have joint
preferences over the schools or daycare centers assigned to their
children and may prefer siblings to attend the same institution
\citep{DurMorrillPhan2022,SunTakenamiMoriwakiTomitaYokoo2023,
SunYokoyamaYokoo2025}. %Such complementarities are difficult to represent using separate assignments for each child, but can be encoded by treating a joint assignment as a finite-side alternative and imposing the relevant institutional capacities through linear constraints. 
Our theorem may
therefore provide a useful analytical foundation for other large
assignment markets involving families or groups.

\subsection*{Related literature}

This paper is related to a growing literature that uses
optimal-transport methods in economic design. 
\citet{daskalakis2013mechanism} develop an optimal-transport framework
for multidimensional mechanism design, and
\citet{daskalakis2017strong} establish strong duality for the
multiple-good monopolist problem. 
\citet{kolesnikov2022beckmann} connect multi-item, multi-bidder auction
design to Beckmann's continuous transportation problem, while
\citet{ashlagi2022price} apply MK duality to waiting-list markets in
which indivisible goods arrive over time. These papers study different
economic environments and exploit problem-specific transport
structures. Our contribution is a general duality and dual-attainment
result for continuum--finite matching problems with finite collections
of linear equality and inequality constraints.

Our analysis is also related to constrained optimal transport.
\citet{zaev2015monge} studies transport plans satisfying additional
linear restrictions while keeping all marginal distributions fixed.
His additional restrictions take the form of equality conditions on
selected averages. By contrast, the finite-side marginal in our model is
determined endogenously, and our formulation permits both equality and
inequality constraints on aggregate assignments. In
capacity-constrained optimal transport, both marginals are again fixed,
and the amount transported along each individual source--destination
route is bounded above; see \citet{korman2015dual}. This pointwise
route-capacity restriction differs from an aggregate capacity on a good
or institution. Neither framework yields the dual-attainment result
needed when the finite-side marginal is endogenous and is restricted by
general linear constraints.

The remainder of the paper is organized as follows.
Section~\ref{sect-MKdual} presents the duality theorem.
Section~\ref{sect-application} develops the two applications.
Section~\ref{sect-conclude} concludes, and
Section~\ref{sect-proof} proves the theorem.

\section{A variant of Monge--Kantorovich duality}
\label{sect-MKdual}

Let $T\subseteq\mathbb{R}^m$ be a compact set of agent types, and let
$\tau$ be a Borel probability measure on $T$. Let $X$ be a nonempty finite
set, endowed with the discrete topology. Its elements represent the goods,
institutions, or bundles to which agents may be matched. Let
$\Phi:T\times X\to\mathbb{R}$ be a continuous surplus function, where
$\Phi(t,x)$ is the surplus generated by matching a type-$t$ agent to
alternative $x$.

A \textbf{matching} $\pi$ is a Borel probability measure on $T\times X$.
For a Borel set $S\subseteq T$ and $x\in X$, the number
$\pi(S\times\{x\})$ is the mass of agents with types in $S$ who are matched
to $x$. Let $\Pi$ denote the set of all matchings. For $\pi\in\Pi$, let
$\pi^T$ denote its marginal distribution on $T$, and identify its marginal
distribution on $X$ with the vector
\[
\pi^X:=\bigl(\pi(T\times\{x\})\bigr)_{x\in X}
\in\mathbb{R}_{\geq 0}^X.
\]
Let $C(T)$ denote the space of real-valued continuous functions on $T$.

We impose linear constraints on the aggregate assignments to the finite side
$X$. Let $A\in\mathbb{R}^{k\times |X|}$ and $a\in\mathbb{R}^k$, and let
$B\in\mathbb{R}^{\ell\times |X|}$ and $b\in\mathbb{R}^{\ell}$. For each
$x\in X$, let $A_x\in\mathbb{R}^k$ and $B_x\in\mathbb{R}^{\ell}$ denote
the columns of $A$ and $B$ indexed by $x$. The formulation includes, for
example, upper and lower capacities, aggregate resource constraints, balance
conditions, and prescribed finite-side marginals. We permit $k=0$ or
$\ell=0$, in which case the corresponding constraints and variables are
omitted. Define
\[
\mathcal{K}:=\mathbb{R}^k_+\times\mathbb{R}^{\ell}.
\]

Consider the following primal and dual problems:
\begin{align*}
\textup{(P)}\qquad
& \sup_{\pi\in\Pi}
    \int_{T\times X}\Phi(t,x)\,d\pi(t,x) \\
& \text{subject to }
    \pi^T=\tau,\qquad A\pi^X\leq a,\qquad B\pi^X=b;
\\[0.5em]
\textup{(D)}\qquad
& \inf_{u\in C(T),\,(p,q)\in\mathcal{K}}
    \left\{
    \int_Tu(t)\,d\tau(t)+a\cdot p+b\cdot q
    \right\} \\
& \text{subject to }
    u(t)+A_x\cdot p+B_x\cdot q\geq\Phi(t,x)
    \quad\text{for $\tau$-a.e. $t\in T$}, \\
& \hspace{8.3em}\text{for every $x\in X$.}
\end{align*}
We write $\operatorname{val}(\mathrm{P})$ and
$\operatorname{val}(\mathrm{D})$ for their optimal values, with the
convention that the supremum over an empty primal feasible set is $-\infty$.
The dual feasible set is nonempty: the function
\[
u_0(t):=\max_{x\in X}\Phi(t,x)
\]
is continuous, and $(u_0,0,0)$ is dual feasible. Hence
$\operatorname{val}(\mathrm{D})<+\infty$.

\begin{theorem}\label{thm-1}
Suppose that at least one of the following conditions holds:
\begin{enumerate}
\item[\textup{(i)}] the primal problem \textup{(P)} is feasible;
\item[\textup{(ii)}] $\operatorname{val}(\mathrm{D})>-\infty$.
\end{enumerate}
Then both problems attain their optimal values, and
\[
\operatorname{val}(\mathrm{P})
=
\operatorname{val}(\mathrm{D})
\in\mathbb{R}.
\]
\end{theorem}

\begin{proof}
See Section~\ref{sect-proof}.
\end{proof}

We next reduce the dual problem to a finite-dimensional minimization. For
$(p,q)\in\mathcal{K}$, define
\[
u_{p,q}(t)
:=
\max_{x\in X}
\bigl\{\Phi(t,x)-A_x\cdot p-B_x\cdot q\bigr\}.
\]
Because $X$ is finite, $u_{p,q}\in C(T)$. The components of $p$ and
$q$ are multipliers on the aggregate constraints, and
$A_x\cdot p+B_x\cdot q$ is the resulting generalized charge for alternative
$x$. In the application of Section~\ref{sect-application}, $u_{p,q}$ becomes
an indirect-surplus function. Define the reduced dual objective
\begin{align}
F(p,q)
&:=
\int_Tu_{p,q}(t)\,d\tau(t)+a\cdot p+b\cdot q,
\qquad (p,q)\in\mathcal{K}.
\label{eq-market-potential}
\end{align}

\begin{lemma}\label{lem:reduce}
The dual value satisfies
\[
\operatorname{val}(\mathrm{D})
=
\inf_{(p,q)\in\mathcal{K}}F(p,q).
\]
If $(p^*,q^*)$ minimizes $F$ over $\mathcal{K}$, then
$(u_{p^*,q^*},p^*,q^*)$ is a solution to \textup{(D)}.
\end{lemma}

\begin{proof}
Fix $(p,q)\in\mathcal{K}$. If $(u,p,q)$ is dual feasible, then
$u(t)\geq u_{p,q}(t)$ for $\tau$-a.e. $t$. Conversely,
$(u_{p,q},p,q)$ is dual feasible. Thus, for fixed $(p,q)$, the smallest
possible value of $\int_Tu\,d\tau$ is attained at $u_{p,q}$. Taking the
infimum over $(p,q)\in\mathcal{K}$ proves the first assertion, and the second
follows immediately.
\end{proof}

Combining Theorem~\ref{thm-1} and Lemma~\ref{lem:reduce}, under either
condition of the theorem,
\[
\operatorname{val}(\mathrm{P})
=
\min_{(p,q)\in\mathcal{K}}F(p,q).
\]
This representation is the economic implication of duality used below: it
represents the infinite-dimensional welfare problem over matchings as a
finite-dimensional minimization over the multipliers attached to the
aggregate constraints. Dual attainment supplies an actual minimizing
multiplier vector, rather than only a minimizing sequence. In
Section~\ref{sect-application}, these multipliers become equilibrium prices
and $F$ specializes to the potential function used to characterize and
compute them.

Before proceeding to application part, we record two basic analytical properties of $F$.

\begin{claim}\label{claim:F-properties}
The function $F$ is finite-valued, convex, and globally Lipschitz continuous
on $\mathcal{K}$.
\end{claim}

\begin{proof}
For each $t\in T$, the map $(p,q)\mapsto u_{p,q}(t)$ is the pointwise
maximum of finitely many affine functions and is therefore convex.
Integration preserves convexity, and $(p,q)\mapsto a\cdot p+b\cdot q$ is
affine. Hence $F$ is convex.

For Lipschitz continuity, write $z:=(p,q)$, $z':=(p',q')$,
$c_x:=(A_x,B_x)$, and $d:=(a,b)$, and use the Euclidean norm on
$\mathbb{R}^{k+\ell}$. For every $t\in T$,
\[
|u_z(t)-u_{z'}(t)|
\leq
\max_{x\in X}|c_x\cdot(z-z')|
\leq
\left(\max_{x\in X}\|c_x\|\right)\|z-z'\|.
\]
Since $\tau$ is a probability measure,
\[
|F(z)-F(z')|
\leq
\left(\max_{x\in X}\|c_x\|+\|d\|\right)\|z-z'\|.
\]
Thus $F$ is globally Lipschitz continuous. Finally, $F$ is finite-valued
because $T$ is compact, $X$ is finite, and $\Phi$ is continuous.
\end{proof}

\section{Applications}
\label{sect-application}

\subsection{Markets with indivisible goods}
\label{subsec:indivisible-goods}

%\subsubsection{Model}

Let $G$ be a finite set of indivisible goods, and let
$X:=\{0,1\}^G$ denote the set of bundles. Let
$T\subseteq\mathbb{R}^{X\setminus\{\mathbf{0}\}}$ be compact, and let
$\tau$ be a Borel probability measure on $T$. For each $t\in T$ and
$x\in X\setminus\{\mathbf{0}\}$, the coordinate $t_x$ is the valuation of
a type-$t$ buyer for bundle $x$; set $t_{\mathbf{0}}:=0$ and
$\Phi(t,x):=t_x$. Let $s\in(0,1)^G$ denote the supply vector.

%As shown by \citet{AzevedoYokote2026}, the original equilibrium-existence proof in \citet{AzevedoWeylWhite2013} contains a gap; see Section~\ref{appsec:AWW-gap} of the Online Appendix. Our duality theorem provides an alternative proof, distinct from the repair proposed by \citet{AzevedoYokote2026}. Relative to their argument, ours is less general because it requires the type space $T$ to be compact. Its advantage is that it also reduces the computation of equilibrium prices to convex minimization (see Theorem \ref{thm:AWW_equiv_char} and the discussion that follows). 

%\subsubsection{Duality and competitive equilibrium}
%\label{sect-equil-existence}

Let $B\in\mathbb{R}^{|G|\times|X|}$ be the incidence matrix whose column
indexed by $x$ is $x$ itself. Consider
\begin{align}
\tag{$\mathrm{P}_{\mathrm{I}}$}\label{eq:P_AWW}
\sup_{\pi\in\Pi}\quad
&\int_{T\times X}\Phi(t,x)\,d\pi(t,x)\\
\text{subject to}\quad
&\pi^T=\tau,\nonumber\\
&B\pi^X=s,\nonumber
\end{align}
and its dual
\begin{align}
\tag{$\mathrm{D}_{\mathrm{I}}$}\label{eq:D_AWW}
\inf_{u,p}\quad
&\int_Tu(t)\,d\tau(t)+p^\top s\\
\text{subject to}\quad
&u(t)+p^\top x\geq\Phi(t,x)
\quad\text{for $\tau$-a.e. $t$ and every $x\in X$},\nonumber\\
&u\in C(T),\quad p\in\mathbb{R}^G.\nonumber
\end{align}
The primal problem is feasible: for example, let $r\in\Delta(X)$ be the
product distribution under which good $g$ is included independently with
probability $s_g$, and take $\pi=\tau\otimes r$. By
Theorem~\ref{thm-1}, both problems have solutions and their optimal values
coincide.

For $p\in\mathbb{R}^G$ and $t\in T$, define 
\begin{align}
D(p,t)
:=
\argmax_{x\in X}\{\Phi(t,x)-p^\top x\}.
\label{eq:demand}
\end{align} 
For $p\in\mathbb{R}^G$, define the indirect surplus and the potential
function by
\[
V_t(p):=\max_{x\in X}\{\Phi(t,x)-p^\top x\},
\qquad
L(p):=\int_TV_t(p)\,d\tau(t)+p^\top s.
\]
By Lemma~\ref{lem:reduce}, the dual problem is equivalent to
$\min_{p\in\mathbb{R}^G}L(p)$, and this minimum is attained.

\begin{definition}\label{def-equil-2}
{\normalfont
A pair $(\pi,p)\in\Pi\times\mathbb{R}^G$ is a \textbf{competitive
equilibrium} if
\[
\pi^T=\tau,
\qquad
B\pi^X=s,
\qquad
x\in D(p,t)
\quad\text{for $\pi$-a.e. }(t,x).
\]
}
\end{definition}

\begin{remark}
{\normalfont 
Definition \ref{def-equil-2} may appear different from the definition of competitive equilibrium in \citet{AzevedoWeylWhite2013}. Our definition represents an allocation by a Borel measure on $T\times X$, whereas \citet{AzevedoWeylWhite2013}
 represent it by a measurable map from types to
lotteries over bundles. 
The two formulations are equivalent;
see Section~\ref{appsec:AWW-equivalence} of the Online Appendix.
} 
\qedo
\end{remark}

\begin{remark}
\label{remark-proof-error} 
{\normalfont 
The gap in the original equilibrium-existence proof of
\citet{AzevedoWeylWhite2013} was first identified in an earlier version of
the present paper: the proof invokes an $L^1$-compactness claim that is
false; see Section~\ref{appsec:AWW-gap} of the Online Appendix.
Following this observation, \citet{AzevedoYokote2026} developed a different
repair of the existence proof that does not rely on duality.
We instead recover equilibrium existence as a direct application of our
duality theorem.
While our duality-based argument assumes that the type space $T$ is compact,
an assumption not imposed by \citet{AzevedoYokote2026}, it yields the
additional benefit of reducing equilibrium-price computation to a
finite-dimensional convex minimization problem; see
Theorem~\ref{thm:AWW_equiv_char} and the discussion that follows.
\qedo
}
\end{remark}

\begin{theorem}
\label{thm:AWW_equiv_char}
Let $\pi^\ast\in\Pi$ and $p^\ast\in\mathbb{R}^G$. The following statements
are equivalent:
\begin{enumerate}
\item $\pi^\ast$ solves \eqref{eq:P_AWW}, and $p^\ast$ minimizes $L$;
\item $(\pi^\ast,p^\ast)$ is a competitive equilibrium.
\end{enumerate}
Consequently, a competitive equilibrium exists.
\end{theorem}

\begin{proof}
For every feasible matching $\pi$ and every $p\in\mathbb{R}^G$,
\begin{align}
L(p)-\int_{T\times X}\Phi(t,x)\,d\pi(t,x)
&=
\int_{T\times X}
\bigl[V_t(p)-\Phi(t,x)+p^\top x\bigr]d\pi(t,x)
\geq0.
\label{eq:AWW-gap}
\end{align}
The equality uses $\pi^T=\tau$ and $B\pi^X=s$. Moreover, equality holds
in \eqref{eq:AWW-gap} if and only if $x\in D(p,t)$ for $\pi$-a.e.
$(t,x)$.

If $\pi^\ast$ solves \eqref{eq:P_AWW} and $p^\ast$ minimizes $L$, strong
duality implies equality in \eqref{eq:AWW-gap}; hence
$(\pi^\ast,p^\ast)$ is a competitive equilibrium. Conversely, if
$(\pi^\ast,p^\ast)$ is a competitive equilibrium, equality holds in
\eqref{eq:AWW-gap}. Since the inequality holds for every feasible $\pi$
and every $p$, this equality implies both that $\pi^\ast$ maximizes the
primal objective and that $p^\ast$ minimizes $L$. Existence follows from
Theorem~\ref{thm-1} and Lemma~\ref{lem:reduce}.
\end{proof}

%\subsubsection{Computation of equilibrium prices}
%\label{sect-AWW-computation}

Theorem~\ref{thm:AWW_equiv_char} identifies equilibrium prices with the
minimizers of the finite-dimensional convex function $L$. Its
subdifferential is
\begin{align}
\partial L(p)
&=
-\int_T\operatorname{co}D(p,t)\,d\tau(t)+\{s\},
\label{eq:AWW-subdiff}
\end{align}
where the integral is the Aumann integral; see
\citet{ioffe1972subdifferentials} or equation~(3) of
\citet{correa2021qualification}. Equivalently,
\[
\partial L(p)
=
\left\{
 s-\int_Td(t)\,d\tau(t)
 \;\middle|\;
 \begin{array}{l}
 d:T\to\mathbb{R}^G\text{ is measurable, and}\\
 d(t)\in\operatorname{co}D(p,t)
 \text{ for $\tau$-a.e. }t
 \end{array}
\right\}.
\]
Thus, every measurable demand selection generates a subgradient equal to
supply minus aggregate demand. Given a measurable, possibly mixed demand
selection
$d^k(t)\in\operatorname{co}D(p^k,t)$,\footnote{For each $k$, the
measurable maximum theorem guarantees the existence of a measurable pure
selection $d^k:T\to X$ satisfying
$d^k(t)\in D(p^k,t)$; see Theorem~18.19 of
\citet{AliprantisBorder2006}. Since
$D(p^k,t)\subseteq\operatorname{co}D(p^k,t)$ and $X$ is finite, this
selection is also admissible in the present formula and is integrable.}
the subgradient iteration is
\[
p^{k+1}
=
p^k+\alpha_k
\left(\int_Td^k(t)\,d\tau(t)-s\right),
\]
where $\alpha_k>0$ is the step size at iteration $k$, which determines the
magnitude of the price adjustment. Thus, prices rise in directions of
excess demand, as in a Walrasian t{\^a}tonnement process. Because $X$ is
finite, the subgradients are uniformly bounded. Under the standard
step-size conditions
\[
\sum_{k=0}^{\infty}\alpha_k=\infty
\qquad\text{and}\qquad
\sum_{k=0}^{\infty}\alpha_k^2<\infty,
\]
standard subgradient-convergence results imply that the distance from
$p^k$ to the set of minimizers of $L$ converges to zero. By
Theorem~\ref{thm:AWW_equiv_char}, this set coincides with the set of
competitive-equilibrium prices. Notably, this conclusion requires no
substitutability assumption: the t{\^a}tonnement process converges to the
set of equilibrium prices even when agents regard some goods as
complements. 

\subsection{Matching markets with couples}
\label{subsec:couples}

Let $H$ be a nonempty finite set of hospitals, let
$H_0:=H\cup\{\varnothing\}$, and let
\[
X:=H_0\times H_0,
\qquad
x^0:=(\varnothing,\varnothing).
\]
An assignment $x=(h_1,h_2)$ specifies the hospitals assigned to the two
members of a couple; the order allows the members to have different roles.
For each $h\in H$ and $x=(h_1,h_2)\in X$, define
\[
a_{h,x}:=\mathbf{1}_{\{h_1=h\}}+\mathbf{1}_{\{h_2=h\}}
\in\{0,1,2\},
\qquad
a_x:=(a_{h,x})_{h\in H},
\]
and let $A:=(a_{h,x})_{h\in H,\,x\in X}$. Thus $a_{h,x}$ is the amount
of hospital $h$'s capacity used by assignment $x$. Let
$s\in\mathbb{R}_+^H$ be the vector of hospital capacities.

Let $T\subseteq\mathbb{R}^{X\setminus\{x^0\}}$ be compact, and let $\tau$
be a Borel probability measure on $T$. For $t\in T$, set
$t_{x^0}:=0$ and define
\[
\Phi(t,x):=t_x
\qquad (t\in T,\ x\in X).
\]
Single agents can be embedded in this formulation by restricting their
positive-surplus assignments to pairs containing at most one nonempty
hospital.

Core nonemptiness can fail in the corresponding finite market. For
example, consider three couples $c_1,c_2,c_3$ and three unit-capacity
hospitals $h_1,h_2,h_3$. Suppose that surplus equals $2$ only for
$(c_1,h_1,h_2)$, $(c_2,h_2,h_3)$, and $(c_3,h_3,h_1)$, and equals $0$
otherwise. At most one productive coalition can be realized. Under any
individually rational feasible payoff for that coalition, one of its two
hospitals receives at most $1$; that hospital can then join the adjacent
unmatched couple and hospital in a blocking coalition. The finite-market
core is therefore empty. In the continuum replica with mass $1/3$ of each
couple type and capacity $1/3$ at each hospital, by contrast, matching mass
$1/6$ of each type to its productive hospital pair and leaving the
remaining couples unmatched supports a core utility profile in which
couples receive zero and each hospital receives $1/3$.

The welfare-maximization problem is
\begin{align}
\tag{$\mathrm{P}_{\mathrm{C}}$}\label{eq:P-couples}
\sup_{\pi\in\Pi}\quad
&\int_{T\times X}\Phi(t,x)\,d\pi(t,x)\\
\text{subject to}\quad
&\pi^T=\tau,\nonumber\\
&A\pi^X\leq s,\nonumber
\end{align}
and its dual is
\begin{align}
\tag{$\mathrm{D}_{\mathrm{C}}$}\label{eq:D-couples}
\inf_{u,p}\quad
&\int_Tu(t)\,d\tau(t)+s\cdot p\\
\text{subject to}\quad
&u(t)+a_x\cdot p\geq\Phi(t,x)
\quad\text{for $\tau$-a.e. $t$ and every $x\in X$},\nonumber\\
&u\in C(T),\quad p\in\mathbb{R}_+^H.\nonumber
\end{align}
The all-unmatched matching is feasible for \eqref{eq:P-couples}; hence
Theorem~\ref{thm-1} implies that \eqref{eq:P-couples} and
\eqref{eq:D-couples} have solutions and equal optimal values.

A \textbf{utility profile} is a pair $(u,v)$, where
$u:T\to\mathbb{R}$ is Borel measurable and $\tau$-integrable and
$v=(v_h)_{h\in H}\in\mathbb{R}^H$. The function $u$ gives the utility of
each couple type, and $v_h$ is the aggregate utility of hospital $h$. Thus
all couples of the same type receive the same utility, even if they are
assigned to different hospital pairs.

A utility profile $(u,v)$ is \textbf{feasible under} a matching $\pi$ if
$\pi$ satisfies the constraints of \eqref{eq:P-couples} and there exist
Borel measurable, $\pi$-integrable functions
$\varphi_h:T\times X\to\mathbb{R}$, $h\in H$, such that
\begin{align*}
&\varphi_h(t,x)=0
&&\text{whenever }a_{h,x}=0,\\
&u(t)+\sum_{h\in H}\varphi_h(t,x)=\Phi(t,x)
&&\text{for $\pi$-a.e. }(t,x),\\
&v_h=\int_{T\times X}\varphi_h(t,x)\,d\pi(t,x)
&&\text{for every }h\in H.
\end{align*}
A utility profile is \textbf{feasible} if it is feasible under some
matching.

To define blocking, represent a coalition of couples by a finite Borel
measure $\mu$ satisfying $0\leq\mu\leq\tau$, meaning that
$\mu(E)\leq\tau(E)$ for every Borel set $E\subseteq T$. This permits a
coalition to contain any fraction of each type. Let $K\subseteq H$ be the
coalition of hospitals, let $K_0:=K\cup\{\varnothing\}$, and let
$X_K:=K_0\times K_0$. Denote by $\Gamma(\mu,K)$ the set of finite Borel
measures $\gamma$ on $T\times X_K$ such that
\[
\gamma^T=\mu
\quad\text{and}\quad
\sum_{x\in X_K}a_{h,x}\gamma(T\times\{x\})\leq s_h
\quad\text{for every }h\in K.
\]
The worth of the coalition is
\[
w(\mu,K)
:=
\sup_{\gamma\in\Gamma(\mu,K)}
\int_{T\times X_K}\Phi(t,x)\,d\gamma(t,x).
\]
The set $\Gamma(\mu,K)$ is nonempty because all couples may be assigned to
$x^0$. We say that $(\mu,K)$ \textbf{blocks} a utility profile $(u,v)$ if
\[
w(\mu,K)
>
\int_Tu(t)\,d\mu(t)+\sum_{h\in K}v_h.
\]

\begin{definition}\label{def:couples-core}
{\normalfont
A feasible utility profile $(u,v)$ is in the \textbf{core} if
\begin{enumerate}
\item $u(t)\geq0$ for $\tau$-a.e. $t$ and $v_h\geq0$ for every $h\in H$;
\item for every $0\leq\mu\leq\tau$ and every $K\subseteq H$,
\[
\int_Tu(t)\,d\mu(t)+\sum_{h\in K}v_h
\geq
w(\mu,K).
\]
\end{enumerate}
}
\end{definition}
The second condition is equivalent to the absence of blocking coalitions
under this transferable-utility criterion.

\begin{theorem}\label{thm:couples-core}
Let $\pi^\ast$ solve \eqref{eq:P-couples}, and let $(u^\ast,p^\ast)$ solve
\eqref{eq:D-couples}. Define
\[
v_h^\ast:=s_hp_h^\ast
\qquad (h\in H).
\]
Then $(u^\ast,v^\ast)$ is feasible under $\pi^\ast$ and belongs to the
core. In particular, the core is nonempty.
\end{theorem}

\begin{proof}
Strong duality gives
\begin{align}
0
&=
\left(\int_Tu^\ast(t)\,d\tau(t)+s\cdot p^\ast\right)
-
\int_{T\times X}\Phi(t,x)\,d\pi^\ast(t,x)\notag\\
&=
\int_{T\times X}
\bigl[u^\ast(t)+a_x\cdot p^\ast-\Phi(t,x)\bigr]
\,d\pi^\ast(t,x)
+
p^\ast\cdot\bigl(s-A\pi^{\ast X}\bigr).
\label{eq:couples-gap}
\end{align}
Both terms on the last line are nonnegative. Therefore,
\[
u^\ast(t)+a_x\cdot p^\ast=\Phi(t,x)
\quad\text{for $\pi^\ast$-a.e. }(t,x),
\]
and
\[
p_h^\ast\bigl(s_h-(A\pi^{\ast X})_h\bigr)=0
\quad\text{for every }h\in H.
\]
For each $h\in H$, define
\[
\varphi_h(t,x):=a_{h,x}p_h^\ast.
\]
These functions satisfy the support and surplus-sharing requirements in the
definition of feasibility. Moreover, complementary slackness implies
\[
\int_{T\times X}\varphi_h(t,x)\,d\pi^\ast(t,x)
=
p_h^\ast(A\pi^{\ast X})_h
=
s_hp_h^\ast
=
v_h^\ast.
\]
Thus $(u^\ast,v^\ast)$ is feasible under $\pi^\ast$.

Individual rationality follows from dual feasibility. Since
$a_{x^0}=0$ and $\Phi(t,x^0)=0$, we have $u^\ast(t)\geq0$ for
$\tau$-a.e. $t$. Also, $v_h^\ast=s_hp_h^\ast\geq0$ for every $h$.

Finally, fix $0\leq\mu\leq\tau$, $K\subseteq H$, and
$\gamma\in\Gamma(\mu,K)$. Since $\mu\leq\tau$, dual feasibility implies
\begin{align*}
\int_{T\times X_K}\Phi(t,x)\,d\gamma(t,x)
&\leq
\int_Tu^\ast(t)\,d\mu(t)
+
\sum_{h\in K}p_h^\ast
\sum_{x\in X_K}a_{h,x}\gamma(T\times\{x\})\\
&\leq
\int_Tu^\ast(t)\,d\mu(t)
+
\sum_{h\in K}s_hp_h^\ast\\
&=
\int_Tu^\ast(t)\,d\mu(t)
+
\sum_{h\in K}v_h^\ast.
\end{align*}
Taking the supremum over $\gamma\in\Gamma(\mu,K)$ proves the core
inequality.
\end{proof}

\section{Concluding remarks}
\label{sect-conclude}

This paper establishes a new Monge--Kantorovich duality theorem for matching problems with a continuum of agents, a finite set of alternatives, and general linear constraints. 
We apply the theorem to two large-market environments. 
In the indivisible-goods model of \citet{AzevedoWeylWhite2013}, it recovers competitive-equilibrium existence and characterizes equilibrium prices as minimizers of a convex potential function, thereby yielding a t\^atonnement-style computational method. 
In the matching-with-couples model, joint assignments use capacity at multiple hospitals and the corresponding finite-market core may be empty. 
In the continuum counterpart, however, primal and dual optimizers support a core utility profile, establishing core nonemptiness.

The two applications illustrate complementary roles of dual attainment. 
It supplies prices that decentralize and facilitate the computation of welfare-maximizing assignments, and it supplies supporting utilities and shadow values that rule out blocking coalitions. 
Because the theorem accommodates general linear constraints on the finite side, the same approach may be useful in other large assignment problems with group preferences or institutional, distributional, and capacity constraints, including the assignment of siblings to schools or daycare centers.

\section{Proof of Theorem \ref{thm-1}} 
\label{sect-proof} 
%Throughout the proof, we write $\text{val(P)}$ to denote the primal optimal value, and $\text{val(D)}$ to denote the dual optimal value. 
We divide the proof into three parts. 
\begin{itemize}
\item Step 1: Proof of \text{val(D)} $\geq$ \text{val(P)} under Condition (i) (Section \ref{sect-proof-step1}). 
\item Step 2: Proof of \text{val(P)} $\geq$ \text{val(D)} and the existence of a primal optimizer (Section \ref{sect-proof-step2}).  
\item Step 3: Proof of the existence of a dual optimizer (Section \ref{sect-proof-step3}). 
\end{itemize} 
The inequality in Step 1 is called {\it weak duality} and is relatively easy to verify.
The proof of Steps 2 is similar to that of Fenchel--Rockafellar duality (see Theorem 1.9 of \cite{villani2021topics}).

The main methodological innovation of our proof appears in Step 3.
Rather than following the standard approach of establishing coercivity of the dual objective or deriving model-specific bounds on prices, we use a different argument. 
\cites{hoffman1952approximate} error bound theorem converts the geometry of the
finite-side polyhedron into an exact penalty parameter, which in turn
identifies a bounded set of economically meaningful multipliers.
\cites{sion1958general} minimax theorem then transfers this finite-dimensional compactness to the dual problem.

\subsection{Proof of \text{val(D)} $\geq$ \text{val(P)} under Condition (i)} \label{sect-proof-step1}
Suppose that Condition (i) holds. 
Take $\pi$ satisfying the constraints of (P). Take $u, p$ satisfying the constraints of (D). Then, 
\begin{align*}
\int \Phi(t, x) d \pi & \leq \int\left\{u(t)+\left(A_x\right)^{\top} \cdot p+\left(B_x\right)^{\top} \cdot q\right\} d \pi \\
& =\int u(t) d \pi+\int\left(A_x\right)^{\top} \cdot p d \pi+\int\left(B_x\right)^{\top} \cdot q d \pi \\
& =\int u(t) d \pi^T+\left(A \pi^X\right)^{\top} \cdot p+\left(B \pi^X\right)^{\top} \cdot q \\
& \leq \int u(t) d \tau+c^{\top} \cdot p+d^{\top} \cdot q, 
\end{align*}
where the first inequality follows from the fact that $u, p$ satisfy the constraints of (D), and the last inequality follows from the fact that $\pi$ satisfies the constraints of (P). 
Taking the supremum of $\pi$ and the infimum of $u, p$, the desired inequality follows.

\subsection{Proof of \text{val(P)} $\geq$ \text{val(D)} and the existence of primal optimizer} 
\label{sect-proof-step2}
If Condition (i) holds, then by \text{val(D)} $\geq$ \text{val(P)} (which is proved in Section \ref{sect-proof-step1}), the dual optimal value is finite. 
This finiteness of the dual optimal value also holds under (ii). 
Therefore, if either Condition (i) or (ii) holds, the dual optimal value is finite, i.e.,  $\text{val(D)}>-\infty$.

Let $C(T\times X)$ denote the set of real-valued continuous functions on $T\times X$. 
We endow this space with the supremum norm, i.e., for each $f \in C(U\times X)$, $\|f\|_{\infty}:=\sup _{(t,x) \in T\times X}|f(t,x)|$. 
We also endow $C(U\times X)\times \mathbb{R}$ with the topology induced from the norm $\|(\psi, r)\|:=\max \left(\|\psi\|_{\infty},|r|\right)$.
We define 
\begin{align*}
F=\{(f, \alpha) \in C(T \times X) \times \mathbb{R} \mid f(u, x) \geq \Phi(u, x), \alpha>0\}. 
\end{align*} 
We also define
\begin{align*}
G=\{(g, \beta) \in C(T \times X) \times \mathbb{R} \mid & \begin{array}{l}
\exists u \in C(T), \exists p \in \mathbb{R}_{+}^k, \exists q \in \mathbb{R}^{\ell} \text { s.t. } \\
g(t, x)=u(t)+\left(A_x\right)^{\top} \cdot p+\left(B_x\right)^{\top} \cdot q
\end{array} \\
& \left.\beta \leq \text{val(D)}-\int u(t) d \tau-c^{\top} \cdot p-d^{\top} \cdot q\right\}. 
\end{align*}
It is clear that $F$ and $G$ are convex sets. It is also clear that the interiors of $F$ is nonempty because any $(f, \alpha)\in C(T \times X) \times \mathbb{R}$ with $f(t,x)>\Phi(t,x)$ for all $(t,x)\in T\times X$ and any $\alpha>0$ form an interior point. Moreover, $F$ and $G$ are disjoint; suppose, for contradiction, that $(f, \alpha)\in F\cap G$. Since $f$ satisfies the membership condition of $G$, $f$ is representable by using $u \in C(T), p \in \mathbb{R}_{+}^k, q \in \mathbb{R}^{\ell}$. By the membership condition of $F$, we have that $u,p,q$ satisfy the dual constraints. Since $\alpha$ satisfies the membership condition of $F$, $\alpha>0$. Together with the fact that $\alpha$ satisfies the membership condition of $G$, 
\begin{align*}
0<\alpha\leq \text{val(D)}-\int u(t) d \tau-c^{\top} \cdot p-d^{\top} \cdot q. 
\end{align*}
We obtain a contradiction to the fact that $\text{val(D)}$ is the dual optimal value. Therefore, $F$ and $G$ are disjoint. We now apply: 

\begin{theorem}[Hahn-Banach (see, e.g., Theorem 5.67 of \cite{AliprantisBorder2006})] 
In any topological vector space, if the interiors of a convex set $F$ is nonempty and is disjoint from another nonempty convex set $G$, then $F$ and $G$ can be properly separated by a nonzero continuous linear functional $\phi$, i.e.,
\begin{align}
\phi(f, \alpha) \geq \phi(g, \beta) \quad \forall(f, \alpha) \in F, \quad \forall(g, \beta) \in G, 
\label{prf-eq1} 
\\
\exists\left(f^{\prime}, \alpha^{\prime}\right) \in F, \quad \exists\left(g^{\prime}, \beta^{\prime}\right) \in G \quad \text{s.t.} \quad \phi\left(f^{\prime}, \alpha^{\prime}\right)>\phi\left(g^{\prime}, \beta^{\prime}\right).
\label{prf-eq2} 
\end{align}
\end{theorem}
Set 
\begin{align*}
\tilde{\phi}(f):=\phi(f, 0) \quad (\forall f \in C(T \times X)), \qquad \hat{\phi}(\alpha)=\phi(\mathbf{0}, \alpha) \quad (\forall \alpha \in \mathbb{R}).
\end{align*} 
Since $\hat{\phi}$ is a linear function on $\mathbb{R}$, there exists $\lambda \in \mathbb{R}$ such that $\hat{\phi}(\alpha)=\lambda \cdot \alpha$ for all $\alpha \in \mathbb{R}$. 
Then, (\ref{prf-eq1}) and (\ref{prf-eq2}) are rewritten as follows: 
\begin{align}
\tilde{\phi}(f)+\lambda \alpha \geq \tilde{\phi}(g)+\lambda \beta \quad \forall(f, \alpha) \in F, \quad \forall(g, \beta) \in G, \label{prf-eq3} \\
\exists\left(f^{\prime}, \alpha^{\prime}\right) \in F, \quad \exists\left(g^{\prime}, \beta^{\prime}\right) \in G \quad \text{s.t.} \quad \tilde{\phi}\left(f^{\prime}\right)+\lambda \alpha^{\prime}>\tilde{\phi}\left(g^{\prime}\right)+\lambda \beta^{\prime}. \label{prf-eq4} 
\end{align}

\begin{claim} \label{cl1} 
It holds that $\lambda\geq 0$.
\end{claim}
\begin{proof}
Suppose $\lambda<0$. By letting $\alpha \rightarrow+\infty$, the LHS of (\ref{prf-eq3}) goes to $-\infty$, contradicting (\ref{prf-eq3}). 
\end{proof} 

\begin{claim} \label{cl2}
$\tilde{\phi}$ is a positive continuous linear functional, i.e., if $f \geq 0$ (i.e., $f(t, x) \geq \mathbf{0}$ for all $(t, x) \in T \times X$), then $\tilde{\phi}(f) \geq 0$.
\end{claim}
\begin{proof}
Suppose, for contradiction, that there exists $f^{\prime} \geq 0$ with $\tilde{\phi}\left(f^{\prime}\right)<0$. By continuity of $\tilde{\phi}$ with respect to the supremum norm topology, there exists $f^{\prime \prime}$ such that $f^{\prime \prime}(t, x)>0$ for all $(t, x) \in T \times X$ and $\tilde{\phi}\left(f^{\prime \prime}\right)<0$. Then, there exists a  sufficiently large $\gamma>0$ scuh that $\left(\gamma^{\prime} f^{\prime \prime}, 1\right) \in F$ for all 
$\gamma^{\prime} \geq \gamma$. Letting $\gamma^{\prime} \rightarrow+\infty$, we have $\tilde{\phi}\left(\gamma^{\prime} f^{\prime \prime}\right)=\gamma^{\prime} \tilde{\phi}\left(f^{\prime \prime}\right) \rightarrow-\infty$.
Then, the LHS of (\ref{prf-eq3}) goes to $-\infty$, contradicting (\ref{prf-eq3}).
\end{proof} 

\begin{claim}\label{cl3} 
It holds that $\lambda >0$.
\end{claim} 
\begin{proof}
Suppose, for contradiction, that $\lambda \leq 0$. By Claim 1, $\lambda=0$. 
By (\ref{prf-eq4}), there exist $f^{\prime} \in C(T\times X)$ and $g^{\prime} \in C(T\times X)$ such that 
%\begin{align*}
$\tilde{\phi}\left(f^{\prime}\right)>\tilde{\phi}\left(g^{\prime}\right)$.
%\end{align*} 
Since $\tilde{\phi}$ is linear, we have $\tilde{\phi}\left(f^{\prime}-g^{\prime}\right)>0$. 
Again by linearity, 
\begin{align}
\tilde{\phi}\left(2f^{\prime}-g^{\prime}\right)=\tilde{\phi}\left(f'+(f^{\prime}-g^{\prime})\right)>\tilde{\phi}\left(f^{\prime}\right).
\label{cl3-eq1} 
\end{align}
We define $h\in C(U\times X)$ by $h:=2f'-g'$. We further define $h^T: T \rightarrow \mathbb{R}$ by
\begin{align*}
h^T(t)=\max \{h(t, x) \mid x \in X\} \quad \forall t\in T. 
\end{align*}
Note that $h^T$ is continuous because $h$ is continuous and $X$ is finite. 
Letting $h'(t,x):=h^T(t)+\left(A_x\right)^{\top} \cdot \mathbf{0}+\left(B_x\right)^{\top} \cdot \mathbf{0}$, we have $(h', \beta)\in G$ for $\beta\leq m-\int h^T(t)d\tau$. Moreover, since $h'(t,x)\geq h(t,x)$ for all $(t,x)\in T\times X$, we have $h'-h\geq \mathbf{0}$. By Claim \ref{cl2}, $\tilde{\phi}(h'-h)\geq 0$. Therefore, 
\begin{align}
\tilde{\phi}\left(h^{\prime}\right)=\tilde{\phi}\left(h+\left(h^{\prime}-h\right)\right) \geq \tilde{\phi}\left(h\right). 
\label{cl3-eq2} 
\end{align}
By combining (\ref{cl3-eq1}) and (\ref{cl3-eq2}), we obtain
\begin{align*}
 \tilde{\phi}\left(h^{\prime}\right)> \tilde{\phi}\left(f^{\prime}\right). 
 \end{align*} 
 This strict inequality yields a contradiction to (\ref{prf-eq3}). 
\end{proof} 
Setting $\phi^{\prime}:=\frac{1}{\lambda} \cdot \tilde{\phi}$ (recall $\lambda>0$ by Claim \ref{cl3}), (\ref{prf-eq3}) is rewritten as 
\begin{align}
\phi^{\prime}(f)+\alpha \geq \phi^{\prime}(g)+\beta \quad \forall(f, \alpha) \in F, \quad \forall(g, \beta) \in G.
\label{prf-eq5} 
\end{align}
Since $\phi'$ is a positive continuous linear functional (recall Claim \ref{cl2}), by Riesz-Markov theorem (see, e.g., Theorem 14.12 of \cite{AliprantisBorder2006}), there exists a Borel measure $\pi^*$ on $T\times X$ such that
\begin{align*}
\phi'(f)=\int f d\pi^* \quad \forall f\in C(T\times X). 
\end{align*}  
Substituting this equation into (\ref{prf-eq5}), 
\begin{align}
\int f d\pi^*+\alpha \geq \int g d\pi^*+\beta \quad \forall(f, \alpha) \in F, \quad \forall(g, \beta) \in G.
\label{prf-eq6} 
\end{align}
Substituting $(\Phi, \alpha)\in F$ with $\alpha>0$ into the left-hand side of (\ref{prf-eq6}) and letting $\alpha \rightarrow 0$, 
\begin{align}
\int \Phi d\pi^* \geq \int g d\pi^*+\beta \quad \forall(g, \beta) \in G.
\label{prf-eq6.5} 
\end{align}
Combining this inequality with the definition of $G$, 
we have the following chain of inequalities: 
\begin{flalign}
\int \Phi d \pi^* \geq \int\left(u(t)+\left(A_x\right)^{\top} \cdot p+\left(B_x\right)^{\top} \cdot q\right) d \pi^*+\text{val(D)}-\int u(t) d \tau-c^{\top} \cdot p-d^{\top} \cdot q  \tag*{}  \\
 \forall u \in C(T), \quad p \in \mathbb{R}_{+}^k, \quad q \in \mathbb{R}^{\ell}, \tag*{} \\
 \int \Phi d \pi^*-\int\left(u(t)+\left(A_x\right)^{\top} \cdot p+\left(B_x\right)^{\top} \cdot q\right) d \pi^*+\int u(t) d \tau+c^{\top} \cdot p+d^{\top} \cdot q \geq \text{val(D)} \tag*{} \\
 \forall u \in C(T), \quad p \in \mathbb{R}_{+}^k, \quad q \in \mathbb{R}^{\ell}, \tag*{} \\
 \int \Phi d \pi^*+\underbrace{\int u(t) d \tau-\int u(t) d \pi^{* T}}_{\text{(I)}}+\underbrace{c^{\top} \cdot p-\left(A \pi^X\right)^{\top} \cdot p}_{\text{(II)}}+\underbrace{d^{\top} \cdot q-\left(B \pi^X\right)^{\top} \cdot q}_{\text{(III)}} \geq \text{val(D)} \tag*{} \\
  \forall u \in C(T), \quad p \in \mathbb{R}_{+}^k, \quad q \in \mathbb{R}^{\ell}. \label{prf-eq7} 
\end{flalign} 
We show that $\pi^*$ satisfies the primal constraints. Suppose, for contradiction, that there exists a Borel subset
$S\subseteq T$ such that $\pi^{*T}(S)\neq \tau(S)$. 
We assume $\pi^{*T}(S)>\tau(S)$ (the other direction can be dealt with analogously). 
By Theorem 15.1 of \cite{AliprantisBorder2006} (p.506), $S$ can be chosen as an open set. For $r\geq 1$, let $u^r:T\rightarrow \mathbb{R}$ be $r$ times the indicator function of $S$, i.e., $u^r(t)=r$ if $t\in S$ and $0$ otherwise.
Since $S$ is open, $u^r=r\mathbf{1}_S$ is lower semicontinuous. Hence, by Baire's theorem (see, e.g., Theorem 4.28 of \cite{Yao2023MK}), there exists a sequence $(u_n^r)\subset C(T)$ such that
\[
0\leq u_n^r\uparrow u^r
\quad\text{pointwise on }T.
\]
By the monotone convergence theorem, the integrals of $u_n^r$ with respect to both $\tau$ and $\pi^{*T}$ converge to the corresponding integrals of $u^r$. Therefore, for $n$ sufficiently large, $\widetilde{u}^r:=u_n^r$ satisfies 
\begin{align*}
\left|\int \tilde{u}^r(t) d \tau-\int u^r(t) d \tau\right|<\varepsilon, \quad \left|\int \tilde{u}^r(t) d \pi^{* T}-\int u^r(t) d \pi^{* T}\right|<\varepsilon, 
\end{align*}
Substituting $\tilde{u}^r$ into $u$ into (\ref{prf-eq7}), 
\begin{align*}
\text{(I)}&=\int \tilde{u}^r(t) d \tau-\int \tilde{u}^r(t) d \pi^{* T} \\
&<\int u^r(t) d \tau+\varepsilon-\int u^r(t) d \pi^{* T}+\varepsilon \\
&= r\cdot (\tau(S)-\pi^{*T}(S))+2\varepsilon<0. 
\end{align*}
Letting $r\rightarrow +\infty$, (I) goes to $-\infty$, which contradicts (\ref{prf-eq7}). 
Therefore, we obtain $\pi^{*T}(S)=\tau(S)$ for all Borel sets $S\subseteq T$. 

If $(A\cdot \pi^{*X})_i>a_i$ for some $i=1, \dots, k$, then by letting $p_i\rightarrow +\infty$, (II) in (\ref{prf-eq7}) goes to $-\infty$, which contradicts (\ref{prf-eq7}). Therefore, we obtain $A\cdot \pi^X\leq a$.  
Similarly, if $(B\cdot \pi^{*X})_j\neq b_j$ for some $j=1, \dots, \ell$, then by letting $p_j\rightarrow +\infty$ (if $(B\cdot \pi^X)_j>b_j$) or $p_j\rightarrow -\infty$  (if $(B\cdot \pi^X)_j< b_j$), (III) in (\ref{prf-eq7}) goes to $-\infty$, which contradicts (\ref{prf-eq7}). Therefore, we obtain $B\cdot \pi^X=b$. 

It follows that $\pi^*$ satisfies the primal constraints.
By \eqref{prf-eq7}, we obtain
\[
\text{val(P)} \geq \int \Phi\, d\pi^* \geq \text{val(D)}.
\]
Moreover, the existence of a matching satisfying the primal constraints implies that Condition (i) holds, which in turn implies the weak duality inequality
\text{(D)} $\geq$ \text{(P)} 
(see Section \ref{sect-proof-step1}).
Combining these inequalities, we have
\[
\text{val(P)}=\int \Phi\, d\pi^*=\text{val(D)},
\]
and hence $\pi^*$ is a solution to the primal problem.
\qed

\subsection{Proof of the existence of dual optimizer} 
\label{sect-proof-step3} 

Let $n:=|X|$ and identify $\R^X$ with $\R^n$.  Define the simplex
\[
\Delta_X
:=
\left\{
 r\in\R_+^X:\ \sum_{x\in X}r_x=1
\right\}
\]
and the set of feasible $X$-marginals
\[
R
:=
\left\{
 r\in\Delta_X:\ Ar\le a,\ Br=b
\right\}.
\tag{1}\label{eq:R-def}
\]
Since (P) has a feasible solution,\footnote{In Step \ref{sect-proof-step2} we proved that, if Condition (i) or (ii) holds, there exists a solution to the primal problem. Therefore, (P) has a feasible solution. } $R$ is nonempty.  It is clear that $R$ is compact.

For $r\in\Delta_X$, let
\[
\Pi(r)
:=
\left\{
 \pi\in\Pi : \ \pi^T=\tau,\ \pi^X=r
\right\}.
\]
This set is nonempty because it contains the product measure
$\tau\otimes r$.  Define
\[
W(r)
:=
\max_{\pi\in\Pi(r)}
\int_{T\times X}\Phi(t,x)\,d\pi(t,x).
\tag{2}\label{eq:W-def}
\]

The maximum in \eqref{eq:W-def} is always attained. Indeed, the maximization problem in \eqref{eq:W-def} is a special case of the primal problem (P)
with no inequality constraints and the single equality constraint $\pi^X=r$. As established in Section \ref{sect-proof-step2}, a maximizer exists whenever a primal-feasible matching exists (i.e., Condition (i) holds).\footnote{Alternatively, one can prove the existence of a maximizer of \eqref{eq:W-def} using weak topology. $T\times X$ is a compact metric space, so $\Pi$ is compact under weak convergence (Prokhorov's theorem). The marginal restrictions defining $\Pi(r)$ are weakly closed, and the map $\pi\mapsto\int\Phi\,d\pi$ is weakly continuous because $\Phi$ is continuous and bounded. The extreme-value theorem implies the existence of a maximizer. See, for example, Theorem~1.4 of \citet{santambrogio2015optimal}
for the corresponding compactness-and-attainment argument for optimal transport problems. }

The function $W$ is concave.  To see this, take $r,r'\in\Delta_X$,
$\lambda\in[0,1]$, and maximizers 
$\pi\in\Pi(r)$ and $\pi'\in\Pi(r')$.  Then
\[
\lambda\pi+(1-\lambda)\pi'
\in
\Pi\bigl(\lambda r+(1-\lambda)r'\bigr),
\]
and therefore
\[
W\bigl(\lambda r+(1-\lambda)r'\bigr)
\ge
\lambda W(r)+(1-\lambda)W(r').
\]

We next establish a Lipschitz bound for $W$.  Set
\[
M:=\|\Phi\|_\infty
=
\max_{(t,x)\in T\times X}|\Phi(t,x)|<\infty.
\]
\begin{claim} \label{claim-Lipschitz}
It holds that 
\[
|W(r)-W(r')|
\le M\|r-r'\|_1
\qquad
\textnormal{for all }r,r'\in\Delta_X.
\tag{4}\label{eq:W-Lipschitz}
\]
\end{claim}
\begin{proof} 
Fix $r,r'\in\Delta_X$.  For each $x\in X$, define
\[
d_x:=\min\{r_x,r'_x\},\qquad
e_x:=(r_x-r'_x)_+,\qquad
f_x:=(r'_x-r_x)_+,
\tag{5}\label{eq:def-d-e-f}
\]
where, for a real number $z$, we define $z_+:=\max\{z,0\}$.  Thus
\[
r_x=d_x+e_x,
\qquad
r'_x=d_x+f_x,
\qquad
e_xf_x=0.
\]
Since both $r$ and $r'$ have total mass one,
\[
\sum_{x\in X}e_x
=
\sum_{x\in X}f_x
=:
\delta
=
\frac12\|r-r'\|_1.
\tag{6}\label{eq:delta}
\]

We now construct a matrix
$\gamma=(\gamma_{xy})_{x,y\in X}$.  If $\delta=0$, set
\[
\gamma_{xy}:=r_x\mathbf{1}_{\{x=y\}}.
\]
If $\delta>0$, set
\[
\gamma_{xy}
:=
d_x\mathbf{1}_{\{x=y\}}
+
\frac{e_xf_y}{\delta}.
\tag{7}\label{eq:gamma-construction}
\]
The second term in \eqref{eq:gamma-construction} distributes the row
excess $e_x$ across the deficit columns in the proportions
$f_y/\delta$.  Because $e_xf_x=0$, no additional mass is put on the
diagonal.  For every $x\in X$,
\[
\sum_{y\in X}\gamma_{xy}
=
d_x+\frac{e_x}{\delta}\sum_{y\in X}f_y
=d_x+e_x=r_x,
\]
and, for every $y\in X$,
\[
\sum_{x\in X}\gamma_{xy}
=
d_y+\frac{f_y}{\delta}\sum_{x\in X}e_x
=d_y+f_y=r'_y.
\]
Thus $\gamma$ has row sums $r$ and column sums $r'$.  Moreover,
\[
\sum_{x\in X}\gamma_{xx}
=
\sum_{x\in X}d_x
=1-\delta,
\qquad
\sum_{x\ne y}\gamma_{xy}=\delta.
\tag{8}\label{eq:gamma-diagonal}
\]

We give an example of matrix $\gamma$ 
for $X=\{1,2,3\}$. Let 
\[
r=\left(\frac13,\frac13,\frac13\right),
\qquad
r'=\left(\frac16,\frac12,\frac13\right).
\]
Then
\[
d=\left(\frac16,\frac13,\frac13\right),
\qquad
e=\left(\frac16,0,0\right),
\qquad
f=\left(0,\frac16,0\right),
\qquad
\delta=\frac16.
\]
Thus the only remaining row excess is $1/6$ in row $1$, and the only
remaining column deficit is $1/6$ in column $2$.  Formula
\eqref{eq:gamma-construction} gives
\[
\gamma
=
\begin{pmatrix}
\frac16 & \frac16 & 0\\[2mm]
0 & \frac13 & 0\\[2mm]
0 & 0 & \frac13
\end{pmatrix}.
\]
Its row sums are $r$, its column sums are $r'$, and its off-diagonal
mass is $1/6=\delta$.

For each $x$ with $r_x>0$, define
\[
K_{xy}:=\frac{\gamma_{xy}}{r_x},
\qquad y\in X.
\]
If $r_x=0$, choose the row $(K_{xy})_{y\in X}$ arbitrarily as a
probability vector.  Then $K$ is a row-stochastic matrix.  Let
$\pi\in\Pi(r)$ attain $W(r)$, and define a probability measure
$\zeta$ on $T\times X\times X$ by
\[
\zeta(E\times\{x\}\times\{y\})
:=
K_{xy}\pi(E\times\{x\})
\]
for every Borel set $E\subseteq T$.  Its $(t,x)$-marginal is $\pi$.
Let $\widetilde\pi$ be its $(t,y)$-marginal, that is,
\[
\widetilde\pi(E\times\{y\})
:=
\sum_{x\in X}K_{xy}\pi(E\times\{x\}).
\tag{9}\label{eq:pi-tilde}
\]
The $T$-marginal of $\widetilde\pi$ is $\tau$, and
\[
\widetilde\pi(T\times\{y\})
=
\sum_{x\in X}r_xK_{xy}
=
\sum_{x\in X}\gamma_{xy}
=r'_y.
\]
Hence $\widetilde\pi\in\Pi(r')$.  In addition,
\[
\zeta\bigl(\{(t,x,y):x\ne y\}\bigr)
=
\sum_{x\ne y}\gamma_{xy}
=
\delta.
\]
It follows that
\begin{align*}
W(r')
&\ge
\int_{T\times X}\Phi(t,y)\,d\widetilde\pi(t,y)\\
&=
\int_{T\times X\times X}\Phi(t,y)\,d\zeta(t,x,y)\\
&=
W(r)
+
\int_{T\times X\times X}
\bigl(\Phi(t,y)-\Phi(t,x)\bigr)\,d\zeta(t,x,y)\\
&\ge
W(r)-2M\delta
=
W(r)-M\|r-r'\|_1, 
\end{align*}
where the second equality follows from the fact that $\pi\in\Pi(r)$ (which is the $(t,x)$-marginal of $\zeta$) attains $W(r)$. 
Interchanging $r$ and $r'$ proves \eqref{eq:W-Lipschitz}.  
\end{proof} 
Claim \ref{claim-Lipschitz} states that $W$ is Lipschitz continuous on $\Delta_X$, which implies that $W$ is continuous on $\Delta_X$.

\begin{claim}[\cites{hoffman1952approximate} error bound theorem]\label{clm:hoffman}
Let $C\in\R^{m_1\times n}$, $c\in\R^{m_1}$,
$D\in\R^{m_2\times n}$, and $d\in\R^{m_2}$.  Suppose that
\[
P:=\{z\in\R^n:Cz\le c,\ Dz=d\}
\]
is nonempty.  Define
\[
\dist_1(z,P)
:=
\inf_{w\in P}\|z-w\|_1.
\]
Then there exists a finite constant $H\ge0$ such that
\[
\dist_1(z,P)
\le
H\left(
\|(Cz-c)_+\|_1+\|Dz-d\|_1
\right)
\qquad
\textnormal{for every }z\in\R^n.
\tag{10}\label{eq:hoffman-general}
\]
Here $(Cz-c)_+$ denotes the componentwise positive part.
\end{claim}

The equality-and-inequality formulation in Claim~\ref{clm:hoffman}
follows from the inequality-only form of Hoffman's theorem by replacing
$Dz=d$ with the two systems $Dz\le d$ and $-Dz\le-d$.  Indeed,
\[
\|(Dz-d)_+\|_1+\|(-Dz+d)_+\|_1
=
\|Dz-d\|_1.
\]

We now show that Claim~\ref{clm:hoffman} applies to
$R$ defined in (\ref{eq:R-def}).  Write $R$ as the solution set of
\[
Ar\le a,
\qquad
-r\le0,
\qquad
Br=b,
\qquad
\1^\top r=1.
\tag{11}\label{eq:R-linear-system}
\]
This is a finite system of linear inequalities and equalities in the
finite-dimensional space $\R^n$.  Its solution set is nonempty, as
noted after \eqref{eq:R-def}.  Thus all prerequisites of
Claim~\ref{clm:hoffman} are satisfied.  Applying the claim with
\[
C=
\begin{pmatrix}
A\\-I_n
\end{pmatrix},
\qquad
c=
\begin{pmatrix}
a\\0
\end{pmatrix},
\qquad
D=
\begin{pmatrix}
B\\\1^\top
\end{pmatrix},
\qquad
d=
\begin{pmatrix}
b\\1
\end{pmatrix},
\]
we obtain a finite constant $H\ge0$ such that, for every
$r\in\Delta_X$,
\[
\dist_1(r,R)
\le
Hh(r),
\qquad
h(r)
:=
\|(Ar-a)_+\|_1+\|Br-b\|_1.
\tag{12}\label{eq:hoffman-application}
\]
The residuals associated with $-r\le0$ and $\1^\top r=1$ vanish
because $r\in\Delta_X$.

Choose any
\[
\rho>MH.
\tag{13}\label{eq:rho-choice}
\]
Fix $r\in\Delta_X$.  There exists $\bar r\in R$ such that
\[
\|r-\bar r\|_1=\dist_1(r,R).
\tag{14}\label{eq:nearest-r}
\]
Indeed, $R$ is nonempty and compact, and the function
$w\mapsto\|r-w\|_1$ is continuous on $R$; hence it attains its minimum.
Using \eqref{eq:W-Lipschitz} and \eqref{eq:hoffman-application}, we get
\begin{align*}
W(r)
&\le W(\bar r)+M\|r-\bar r\|_1\\
&\le \text{val(P)}+MHh(r).
\end{align*}
Consequently,
\[
W(r)-\rho h(r)
\le
\text{val(P)}-(\rho-MH)h(r)
\le \text{val(P)}.
\tag{15}\label{eq:penalty-upper-bound}
\]
Since $R$ is compact and $W$ is continuous, there exists
$r^*\in R$ with $W(r^*)=\max_{r\in R}W(r)$.  
It also holds that 
\[
\text{val(P)}=W(r^*).
\label{eq:primal-marginal}
\]
Indeed, for a solution $\hat{\pi}$ to the primal problem (whose existence is established in Section \ref{sect-proof-step2}), letting $\hat{r}:=\hat{\pi}^X\in R$, we have $v^*=W(\hat{r})\leq \max_{r\in R}W(r)=W(r^*)$. 
Conversely, since $r^* \in R$, every matching in $\Pi(r^*)$ is primal feasible, which implies $v^*\geq W(r^*)$.

Combining the above displayed equality with (\ref{eq:penalty-upper-bound}), we obtain $h(r^*)=0$ 
and \eqref{eq:penalty-upper-bound} is exact in the sense that
\[
\text{val(P)}
=
\max_{r\in\Delta_X}
\bigl\{W(r)-\rho h(r)\bigr\}.
\tag{16}\label{eq:exact-penalty}
\]

Define
\[
Q_\rho
:=
[0,\rho]^k\times[-\rho,\rho]^\ell
\]
and, for $r\in\Delta_X$ and $(p,q)\in Q_\rho$, define
\[
\mathcal{L}(r,p,q)
:=
W(r)+p\cdot(a-Ar)+q\cdot(b-Br).
\tag{17}\label{eq:Lagrangian}
\]
For each fixed $r\in\Delta_X$, minimization over the coordinates of
$p$ and $q$ gives
\begin{align*}
\min_{0\le p_i\le\rho}
 p_i\bigl(a_i-(Ar)_i\bigr)
&=
-\rho\bigl((Ar)_i-a_i\bigr)_+,\\
\min_{-\rho\le q_j\le\rho}
 q_j\bigl(b_j-(Br)_j\bigr)
&=
-\rho\bigl|b_j-(Br)_j\bigr|.
\end{align*}
Hence
\[
\min_{(p,q)\in Q_\rho}\mathcal{L}(r,p,q)
=
W(r)-\rho h(r).
\tag{18}\label{eq:penalty-as-multipliers}
\]
Combining \eqref{eq:exact-penalty} and
\eqref{eq:penalty-as-multipliers},
\[
\text{val}(P)
=
\max_{r\in\Delta_X}
\min_{(p,q)\in Q_\rho}
\mathcal{L}(r,p,q).
\tag{19}\label{eq:max-min-before-sion}
\]

\begin{claim}[\cites{sion1958general} minimax theorem]\label{clm:sion}
Let $C$ and $D$ be nonempty convex subsets of Hausdorff topological
vector spaces, and suppose that $C$ is compact.  Let
$f:C\times D\to\R$ satisfy the following conditions:
\begin{enumerate}[label=\textnormal{(\roman*)}]
\item for every $y\in D$, the function $x\mapsto f(x,y)$ is upper
semicontinuous and quasi-concave on $C$;
\item for every $x\in C$, the function $y\mapsto f(x,y)$ is lower
semicontinuous and quasi-convex on $D$.
\end{enumerate}
Then
\[
\sup_{x\in C}\inf_{y\in D}f(x,y)
=
\inf_{y\in D}\sup_{x\in C}f(x,y).
\tag{20}\label{eq:sion-general}
\]
\end{claim}

We verify every prerequisite of Claim~\ref{clm:sion}.  Take
\[
C=\Delta_X,
\qquad
D=Q_\rho,
\qquad
f=\mathcal{L}.
\]
Both $\Delta_X$ and $Q_\rho$ are nonempty, compact, convex subsets of
finite-dimensional Euclidean spaces, which are Hausdorff topological
vector spaces.  For each fixed $(p,q)\in Q_\rho$, the map
$r\mapsto\mathcal{L}(r,p,q)$ is continuous because $W$ is Lipschitz,
and it is concave because $W$ is concave and the remaining terms are
linear in $r$.  It is therefore upper semicontinuous and
quasi-concave.  For each fixed $r\in\Delta_X$, the map
$(p,q)\mapsto\mathcal{L}(r,p,q)$ is affine and continuous.  It is
therefore lower semicontinuous and quasi-convex.  Thus all hypotheses
of Claim~\ref{clm:sion} hold.  Since $\Delta_X$ and $Q_\rho$ are compact and
$\mathcal{L}$ is continuous, the suprema and infima in this application
are attained.  Therefore,
\[
\max_{r\in\Delta_X}
\min_{(p,q)\in Q_\rho}\mathcal{L}(r,p,q)
=
\min_{(p,q)\in Q_\rho}
\max_{r\in\Delta_X}\mathcal{L}(r,p,q).
\tag{21}\label{eq:sion-application}
\]
Combining \eqref{eq:max-min-before-sion} and
\eqref{eq:sion-application},
\[
\text{val(P)}
=
\min_{(p,q)\in Q_\rho}
\max_{r\in\Delta_X}\mathcal{L}(r,p,q).
\tag{22}\label{eq:after-sion}
\]

It remains to identify the inner maximum.  Fix
$(p,q)\in Q_\rho$, and write
\[
g_x(t)
:=
\Phi(t,x)-A_x\cdot p-B_x\cdot q, 
\]
for each $x\in X$. 
Using the definition of $W$ and the fact that every matching with
$T$-marginal $\tau$ has some $X$-marginal in $\Delta_X$, we obtain
\begin{align*}
\max_{r\in\Delta_X}\mathcal{L}(r,p,q)
&=
a\cdot p+b\cdot q\\
&\quad+
\max_{r\in\Delta_X}
\left\{
W(r)-\sum_{x\in X}r_x(A_x\cdot p+B_x\cdot q)
\right\}\\
&=
a\cdot p+b\cdot q
+
\max_{\pi\in\Pi:\,\pi^T=\tau}
\int_{T\times X}g_x(t)\,d\pi(t,x).
\tag{23}\label{eq:inner-max-matching}
\end{align*}
We remark that the maximum in the last line is attainable by choosing a matching $\pi\in W(r^*)$, where $r^*$ is a solution to $\max W(r)$.\footnote{As noted after (\ref{eq:W-def}), for any $r\in \Delta_X$, there is a solution to $\max W(r)$.}  
For every matching in the last line,
\[
\int_{T\times X}g_x(t)\,d\pi(t,x)
\le
\int_T\max_{x\in X}g_x(t)\,d\tau(t).
\tag{24}\label{eq:pointwise-upper}
\]
Equality in \eqref{eq:pointwise-upper} is attainable.  To see this
formally, enumerate $X=\{1,\ldots,n\}$ and choose the smallest-index
maximizer.  Specifically, define
\[
E_i
:=
\left(\bigcap_{j=1}^n\{t:g_i(t)\ge g_j(t)\}\right)
\cap
\left(\bigcap_{j<i}\{t:g_i(t)>g_j(t)\}\right).
\]
Each $E_i$ is Borel because the functions $g_i-g_j$ are continuous,
and the sets $E_1,\ldots,E_n$ form a measurable partition of $T$.
Define $\widehat x(t)=i$ for $t\in E_i$.  Then $\widehat x$ is Borel
measurable and
\[
g_{\widehat x(t)}(t)=\max_{x\in X}g_x(t)
\qquad\textnormal{for every }t\in T.
\]
We define $\hat{\pi}\in \Pi$ by 
\[
\widehat\pi(E)
=
\tau\bigl(
\{t\in T:(t,\widehat x(t))\in E\}
\bigr), 
\]
for every Borel subset $E\subseteq T\times X$. 
Then, $\widehat\pi$ has $T$-marginal $\tau$ and attains equality in
\eqref{eq:pointwise-upper}.  Therefore,
\begin{align*}
\max_{r\in\Delta_X}\mathcal{L}(r,p,q)
&=
\int_T
\max_{x\in X}
\bigl\{\Phi(t,x)-A_x\cdot p-B_x\cdot q\bigr\}
\,d\tau(t)+a\cdot p+b\cdot q\\
&=F(p,q), 
\tag{25}\label{eq:inner-max-F}
\end{align*}
where $F$ is the same function as that defined in (\ref{eq-market-potential}). 
Substituting \eqref{eq:inner-max-F} into \eqref{eq:after-sion},
\[
\text{val(P)}=\min_{(p,q)\in Q_\rho}F(p,q).
\tag{26}\label{eq:F-min-compact}
\]

Since $Q_\rho$ is compact, there exists
$(p^*,q^*)\in Q_\rho$ satisfying
\[
F(p^*,q^*)=\text{val(P)}.
\tag{27}\label{eq:F-optimum}
\]
Then, we obtain 
\begin{align*}
F(p^*,q^*)=\text{val(P)}\leq \text{val(D)}=\inf_{(p,q)\in\mathcal{K}} F(p,q), 
\end{align*}
where the inequality is the weak duality established in (\ref{sect-proof-step1}) and the second equality follows from Lemma~\ref{lem:reduce}. Again by Lemma~\ref{lem:reduce}, 
$(u_{p^*,q^*},p^*,q^*)$ is a solution to the dual problem. 
\qed

\clearpage
\appendix
\section{Online Appendix}
\label{app:online}

\subsection{Equivalence of the two equilibrium formulations}
\label{appsec:AWW-equivalence}
We show that the definition of competitive equilibrium given in Definition \ref{def-equil-2} is equivalent to the definition in \citet{AzevedoWeylWhite2013}. 

Let $\Delta(X)$ denote the set of probability
distributions over $X$. 
An \textbf{allocation} is a measurable map
$\mathbf{x}:T\to\Delta(X)$. Its induced expected demand is
\[
\widetilde{\mathbf{x}}(t)
:=
\sum_{x\in X}\mathbf{x}(t)_x x
\in[0,1]^G.
\]
Recall that, for $p\in\mathbb{R}^G$ and $t\in T$, 
\begin{align}
D(p,t)
:=
\argmax_{x\in X}\{\Phi(t,x)-p^\top x\}.
\label{eq:demand-2}
\end{align}
Let $\overline{D}(p,t)\subseteq\Delta(X)$ be the set of lotteries whose
supports are contained in $D(p,t)$. 

\citet{AzevedoWeylWhite2013} use the following equilibrium
notion.
\begin{definition}\label{def-equil-1}
{\normalfont
A \textbf{competitive equilibrium} is a price--allocation pair
$(p,\mathbf{x})$ such that
\[
\int_T\widetilde{\mathbf{x}}(t)\,d\tau(t)=s
\quad\text{and}\quad
\mathbf{x}(t)\in\overline{D}(p,t)
\quad\text{for $\tau$-a.e. }t.
\]
}
\end{definition}

Definition~\ref{def-equil-1} represents an assignment by a measurable map
from types to lotteries over bundles, whereas
Definition~\ref{def-equil-2} represents it by a probability measure on
$T\times X$. Because $X$ is finite, the two representations are equivalent.

\begin{proposition}\label{prop:AWW-formulations-equivalent}
The following statements hold.
\begin{enumerate}
\item Every allocation $\mathbf{x}:T\to\Delta(X)$ induces a unique Borel
probability measure $\pi_{\mathbf{x}}$ on $T\times X$ satisfying
\begin{align}
\pi_{\mathbf{x}}(V\times\{x\})
=
\int_V\mathbf{x}(t)_x\,d\tau(t)
\label{eq:allocation-to-measure}
\end{align}
for every Borel set $V\subseteq T$ and every $x\in X$. Moreover,
$\pi_{\mathbf{x}}^T=\tau$.

\item Conversely, every Borel probability measure $\pi$ on $T\times X$
with $\pi^T=\tau$ is induced by an allocation $\mathbf{x}_\pi$ through
\eqref{eq:allocation-to-measure}. This allocation is unique up to
$\tau$-a.e. equality.

\item Under this correspondence,
\[
B\pi_{\mathbf{x}}^X
=
\int_T\widetilde{\mathbf{x}}(t)\,d\tau(t),
\]
and, for every $p\in\mathbb{R}^G$,
\[
x\in D(p,t)
\quad\text{for $\pi_{\mathbf{x}}$-a.e. }(t,x)
\]
if and only if
\[
\mathbf{x}(t)\in\overline{D}(p,t)
\quad\text{for $\tau$-a.e. }t.
\]
Consequently, Definitions~\ref{def-equil-1} and~\ref{def-equil-2} are
equivalent.
\end{enumerate}
\end{proposition}

\begin{proof}
Let $\mathbf{x}:T\to\Delta(X)$ be an allocation. For every Borel set
$E\subseteq T\times X$, define
\[
\pi_{\mathbf{x}}(E)
:=
\sum_{x\in X}
\int_T
\mathbf{1}_{E}(t,x)\mathbf{x}(t)_x\,d\tau(t).
\]
Because $X$ is finite, this defines a Borel probability measure on
$T\times X$. For every Borel set $V\subseteq T$,
\[
\pi_{\mathbf{x}}(V\times X)
=
\int_V\sum_{x\in X}\mathbf{x}(t)_x\,d\tau(t)
=
\tau(V),
\]
so $\pi_{\mathbf{x}}^T=\tau$. Equation
\eqref{eq:allocation-to-measure} follows directly from the definition,
and it uniquely determines $\pi_{\mathbf{x}}$ because $X$ is finite.

Conversely, let $\pi$ be a Borel probability measure on $T\times X$ with
$\pi^T=\tau$. For each $x\in X$, define a finite Borel measure $\pi_x$ on
$T$ by
\[
\pi_x(V):=\pi(V\times\{x\}).
\]
If $\tau(V)=0$, then
\[
0
=
\pi^T(V)
=
\sum_{y\in X}\pi_y(V),
\]
and hence $\pi_x(V)=0$ for every $x\in X$. Thus $\pi_x\ll\tau$. By the
Radon--Nikodym theorem, there exists a measurable function
$\mathbf{x}_\pi(\cdot)_x:T\to\mathbb{R}_{\geq0}$ such that
\[
\pi_x(V)
=
\int_V\mathbf{x}_\pi(t)_x\,d\tau(t)
\qquad\text{for every Borel }V\subseteq T.
\]
Moreover, for every Borel set $V\subseteq T$,
\[
\int_V\sum_{x\in X}\mathbf{x}_\pi(t)_x\,d\tau(t)
=
\sum_{x\in X}\pi_x(V)
=
\pi^T(V)
=
\tau(V).
\]
It follows that
$\sum_{x\in X}\mathbf{x}_\pi(t)_x=1$ for $\tau$-a.e. $t$.
After redefining $\mathbf{x}_\pi$ on a null set, if necessary, we obtain a
measurable map from $T$ to $\Delta(X)$. The Radon--Nikodym derivatives are
unique $\tau$-a.e., proving uniqueness of the allocation up to null sets.
The identities above also imply $\pi_{\mathbf{x}_\pi}=\pi$.

Since the column of $B$ indexed by $x$ is $x$, we have
\begin{align*}
B\pi_{\mathbf{x}}^X
&=
\sum_{x\in X}x\,\pi_{\mathbf{x}}(T\times\{x\})\\
&=
\sum_{x\in X}x\int_T\mathbf{x}(t)_x\,d\tau(t)\\
&=
\int_T\sum_{x\in X}\mathbf{x}(t)_x x\,d\tau(t)
=
\int_T\widetilde{\mathbf{x}}(t)\,d\tau(t).
\end{align*}
Thus market clearing is identical in the two representations.

Finally, fix $p\in\mathbb{R}^G$. Because $X$ is finite and
$\Phi(\cdot,x)$ is continuous for every $x$, the set
\[
N_p
:=
\{(t,x)\in T\times X:x\notin D(p,t)\}
\]
is Borel. By the definition of $\pi_{\mathbf{x}}$,
\[
\pi_{\mathbf{x}}(N_p)
=
\sum_{x\in X}
\int_{\{t:\,x\notin D(p,t)\}}
\mathbf{x}(t)_x\,d\tau(t).
\]
All terms on the right-hand side are nonnegative. Hence
$\pi_{\mathbf{x}}(N_p)=0$ if and only if, for every $x\in X$,
$\mathbf{x}(t)_x=0$ for $\tau$-a.e. $t$ such that
$x\notin D(p,t)$. Since $X$ is finite, this is equivalent to
$\mathbf{x}(t)\in\overline{D}(p,t)$ for $\tau$-a.e. $t$. Combining this
observation with the market-clearing identity proves the equivalence of
Definitions~\ref{def-equil-1} and~\ref{def-equil-2}.
\end{proof}

\subsection{Failure of $L^1$-compactness in the original existence proof}
\label{appsec:AWW-gap}

In proving that the aggregate-demand correspondence has a closed graph,
\citet[p.~286]{AzevedoWeylWhite2013} invoke the claim that the set of
allocations is compact under the $L^1$ norm and therefore that a convergent
subsequence of allocations may be selected. We show that this compactness
claim is false, even when the type distribution has compact support.

Let
\[
\mathcal{A}
:=
\bigl\{
\mathbf{x}:T\to\Delta(X)
\mathrel{\big|}
\mathbf{x}\text{ is measurable}
\bigr\}
\]
be the set of allocations, endowed with the $L^1$ metric
\[
\|\mathbf{x}-\mathbf{y}\|_1
:=
\sum_{z\in X}
\int_T
\bigl|\mathbf{x}(t)_z-\mathbf{y}(t)_z\bigr|\,d\tau(t).
\]
For each $p\in\mathbb{R}^G$, define aggregate demand by
\[
\mathcal{D}(p)
:=
\left\{
\int_T\widetilde{\mathbf{x}}(t)\,d\tau(t)
\ \middle|\
\mathbf{x}\in\mathcal{A},\ 
\mathbf{x}(t)\in\overline{D}(p,t)
\text{ for $\tau$-a.e. }t
\right\}.
\]

\begin{proposition}\label{prop:AWW-L1-noncompact}
There exist sequences $p^n\to p$ and $d^n\to d$, together with allocations
$\mathbf{x}^n\in\mathcal{A}$ satisfying
\[
d^n
:=
\int_T\widetilde{\mathbf{x}}^n(t)\,d\tau(t),
\qquad
d^n\in\mathcal{D}(p^n)
\quad\text{for every }n,
\]
such that $\{\mathbf{x}^n\}_{n\geq1}$ has no $L^1$-convergent
subsequence. Consequently, $\mathcal{A}$ is not compact under the $L^1$
metric.
\end{proposition}

\begin{proof}
Let $G=\{1,2\}$ and
\[
X=\{\mathbf{0},e_1,e_2,e_{12}\},
\qquad
\mathbf{0}:=(0,0),
\]
\[
e_1:=(1,0),
\qquad
e_2:=(0,1),
\qquad
e_{12}:=(1,1).
\]
For each $r\in[0,1]$, define a type $t^r$ by
\[
t^r_{e_1}=t^r_{e_2}=t^r_{e_{12}}:=1+r,
\qquad
t^r_{\mathbf{0}}:=0,
\]
and let
\[
T^{=}:=\{t^r:r\in[0,1]\}.
\]
The set $T^{=}$ is compact. For this example, take $T:=T^{=}$. Let
$\lambda$ denote Lebesgue measure on $[0,1]$, and let $\tau$ be its
pushforward under the map $r\mapsto t^r$. We henceforth identify $T$
with $[0,1]$ through this map.

Set
\[
p^n:=p:=\left(\frac12,\frac12\right)
\qquad\text{for every }n.
\]
For every $r\in[0,1]$,
\[
t^r_{e_1}-p^\top e_1
=
t^r_{e_2}-p^\top e_2
=
\frac12+r
>
r
=
t^r_{e_{12}}-p^\top e_{12},
\]
and $\frac12+r>0=t^r_{\mathbf{0}}-p^\top\mathbf{0}$. Therefore,
\[
D(p,t^r)=\{e_1,e_2\}
\qquad\text{for every }r\in[0,1].
\]
Let $\delta_z\in\Delta(X)$ denote the degenerate lottery on $z\in X$.
For each $n\geq1$, define $\mathbf{x}^n:T^{=}\to\Delta(X)$ by
\[
\mathbf{x}^n(t^r)
:=
\begin{cases}
\delta_{e_1},
& \text{if $\lfloor 2^n r\rfloor$ is odd},\\
\delta_{e_2},
& \text{if $\lfloor 2^n r\rfloor$ is even},
\end{cases}
\qquad r\in[0,1),
\]
and set $\mathbf{x}^n(t^1):=\delta_{e_2}$. Each $\mathbf{x}^n$ is measurable
and assigns only demanded bundles. Moreover, half of the type space is
assigned to each of $e_1$ and $e_2$, so
\[
d^n
:=
\int_T\widetilde{\mathbf{x}}^n(t)\,d\tau(t)
=
\left(\frac12,\frac12\right)
\qquad\text{for every }n.
\]
Thus $(p^n,d^n)$ is a constant, and hence convergent, sequence with
$d^n\in\mathcal{D}(p^n)$.

Fix $m>n$. On each dyadic interval
$[k2^{-n},(k+1)2^{-n})$, the allocation $\mathbf{x}^n$ is constant,
whereas $\mathbf{x}^m$ alternates between $\delta_{e_1}$ and
$\delta_{e_2}$ on the $2^{m-n}$ subintervals of length $2^{-m}$.
Consequently,
\[
\lambda\bigl(
\{r\in[0,1]:\mathbf{x}^m(t^r)\neq\mathbf{x}^n(t^r)\}
\bigr)
=
\frac12.
\]
Whenever the two allocations differ, the $\ell^1(X)$ distance between
the corresponding degenerate lotteries is $2$. Hence
\[
\|\mathbf{x}^m-\mathbf{x}^n\|_1
=
2\cdot\frac12
=
1.
\]
Every subsequence is therefore separated by distance $1$ and cannot be
Cauchy. Since the $L^1$ topology is metrizable, $\mathcal{A}$ cannot be
compact.
\end{proof}

Proposition~\ref{prop:AWW-L1-noncompact} does not contradict the
closed-graph conclusion itself. It shows that the $L^1$-compactness
argument invoked in the original proof does not establish that conclusion.

\bibliography{new-duality}

\end{document}